\documentclass[11pt]{article}

\usepackage{amsfonts,amssymb,amsmath,latexsym,xfrac,ae,aecompl}

\usepackage{graphicx}

\newcommand{\FFF}{\vspace*{\bigskipamount}}

\newcommand{\BBB}{\vspace*{-\bigskipamount}}

\newcommand{\cO}{\mathcal{O}}

\newcommand{\Paragraph}[1]{\BBB\paragraph{#1}}
\newcommand{\remove}[1]{}

\newlength{\figurewidth}
\newlength{\figureheight}
\newlength{\captionwidth}
\newcommand{\qed}{\hfill $\square$ \smallbreak}
\newenvironment{proof}{\noindent{\bf Proof:}}{\qed}

\newtheorem{theorem}{Theorem}

\newtheorem{proposition}{Proposition}
\newtheorem{conjecture}{Conjecture}

\begin{document}

\addtolength{\baselineskip}{1pt}        
   
\addtolength{\parskip}{2pt}

\title{Broadcasting on Adversarial Multiple Access Channels \vfill}

\author{Bader A. Aldawsari \footnotemark[1]
	\and
	Bogdan S. Chlebus \footnotemark[2]
	\and
	Dariusz R. Kowalski \footnotemark[2]}

\footnotetext[1]{	Department of Computer Science and Engineering, University of Colorado Denver, Denver, Colorado, USA.}

\footnotetext[2]{	School of Computer and Cyber Sciences, Augusta University, Augusta, Georgia, USA.}

\date{}

\maketitle

\vfill


\begin{abstract}
We study deterministic distributed algorithms for broadcasting on multiple-access channels.
Packet injection is modeled by leaky-bucket adversaries.
There is a fixed set of stations attached to a channel.
Additional features of the model of communication include an upper bound on the number of stations activated in a round, an individual injection rate, and randomness in generating and injecting packets.
We demonstrate that some broadcast algorithms designed for ad-hoc channels have bounded latency for increased ranges of injection rates than in ad-hoc channels when executed on channels with a fixed number of stations against adversaries that can activate at most one station per round.
Individual injection rates are shown to impact latency, as compared to the model of general leaky bucket adversaries.
Outcomes of experiments are given that compare the performance of broadcast algorithms against randomized adversaries.
The experiments include deterministic algorithms and randomized backoff algorithms.

\vfill

\noindent
\textbf{Key words:}
multiple-access channel, 
distributed broadcast,
adversarial packet injection,
queue stability,
packet latency.
\end{abstract}

\vfill

~

\thispagestyle{empty}

\setcounter{page}{0}


\newpage

\section{Introduction}

\label{sec:introduction}

We consider distributed broadcast algorithms on multiple access channels.
The goal is to compare their performance with respect to packet latency and maximum queue size when packet generation is constrained by adversarial models.
We investigate how the properties of channels and kinds of adversaries and classes of algorithms interplay among themselves to impact the performance of broadcasting.

Adversarial queuing is a methodology that can capture stability of communication without any statistical assumptions about traffic.
It can provide a framework for worst-case bounds on performance of deterministic distributed algorithms.
This approach was proposed by Borodin et al.~\cite{BorodinKRSW-JACM01} in their study of routing algorithms in store-and-forward networks, and continued by Andrews et al.~\cite{AndrewsAFLLK-JACM01}.
Berger et al.~\cite{BergerKS14} studied engineering aspects of adversarial queuing modeling packet injection.
The latency of routing in adversarial queueing was studied by Aiello et al.~\cite{AielloKOR-JCSS00}, Andrews et al.~\cite{AndrewsFGZ-JACM05}, and Rosen and Tsirkin~\cite{RosenT-TCSy06}.
A survey of adversarial contention resolution was given by Banicescu et al.~\cite{BanicescuCGY2024}.

Multiple-access channels model contention occurring in local-area networks utilizing Ethernet protocols, see Metcalfe and Boggs~\cite{MetcalfeB76}.
Broadcasting algorithms for multiple-access channels need to resolve contention for access and typically use randomness in their design; see Chang et al.~\cite{ChangJP19}, Zhou et al.~\cite{ZhouCY22}, and a survey by Chlebus~\cite{Chlebus-chapter-2001}.
Throughput of channels using Ethernet protocols was investigated by Bianchi~\cite{Bianchi00} when the network is saturated, in the sense that each node has always a packet to transmit.
Kwak et al.~\cite{KwakSM05} studied saturated throughput of variants of backoff broadcast  on multiple-access channels. 

Stability of randomized communication algorithms on multiple-access channels can be considered in  the queue-free model, in which a  packet gets associated with a new station at the time of injection; see Chlebus~\cite{Chlebus-chapter-2001} for a survey of this topic.
The model of a fixed set of stations each with its own queue  appears to have a stabilizing effect on randomized broadcasting.
This was demonstrated by H\aa stad et al.~\cite{HastadLR-SICOMP96}, Al-Ammal et al.~\cite{Al-AmmalGM-TCSy01} and  Goldberg et al.~\cite{GoldbergJKP04}, who investigated the range of injection rates for which the binary exponential backoff is stable, as a function of the number of stations.
The efficiency of backoff-related algorithms for the batched-arrival and slack-constrained models of packet creation was studied by Bender et al.~\cite{BenderKPY18}, Bender et al.~\cite{BenderFGY19}, Bender et al.~\cite{BenderKKP20}, Agrawal et al.~\cite{AgrawalBFGY20}, and Anderton et al.~\cite{AndertonCY21}.

Broadcasting in multiple-access channels with queue-free stations in the framework of adversarial queuing was first studied by Bender et al.~\cite{BenderFHKL-SPAA05}.
Deterministic distributed broadcast performed by stations with queues was introduced in the adversarial setting by Chlebus et al.~\cite{ChlebusKR-TALG12}.
The maximum throughput in such a setting was studied in Chlebus et al.~\cite{ChlebusKR-DC09}; that paper demonstrated that the ultimate throughput of~$1$ was achievable.
Anantharamu et al.~\cite{AnantharamuCR-TCS17} extended the work on throughput~$1$ to adversarial models with individual injection rates.
Anantharamu and Chlebus~\cite{AnantharamuC15} developed deterministic distributed broadcast algorithms for ad-hoc $1$-activating channels; that work also demonstrated that no broadcast algorithm can be universal in such channels.
Anantharamu et al.~\cite{AnantharamuCKR-JCSS19} investigated the latency of adversarial broadcasting by deterministic algorithms.
Hradovich et al.~\cite{HradovichKK21} considered deterministic broadcasting on adversarial multiple-access channels when the adversary has to continually maintain the maximum allowed injection rate.   
Energy constraints in multiple-access channels were studied by Chlebus et al.~\cite{ChlebusHJKK-SPAA19} for adversarial routing and by Jiang and Zheng~\cite{JiangZ22}  for collision resolution with batch arrivals.
The impact of jamming on collision resolution was studied by Anantharamu et al.~\cite{AnantharamuCKR-JCSS19} for broadcasting with adversarial packet arrivals and by Chen et al.~\cite{ChenJZ21}.

\Paragraph{An overview of the document.}

Section~\ref{sec:technical-preliminaries} reviews the relevant properties of shared channels and communication algorithms for broadcasting on such channels.
Section~\ref{sec:adversarial-packet-injection} discusses adversarial models of packets injection for dynamic broadcasting. 
We propose a process of injecting packets  called ``four step bucket process.''
This process is shown in Proposition~\ref{pro:one} to be equivalent to the leaky-bucket packet-generation model.
We use the four step bucket model to introduce randomized leaky-bucket adversaries.
Proposition~\ref{pro:simulating-adversary} shows that any sequence of numbers of packets generated by a leaky-bucket adversary of injection rate $\rho$ and burstiness~$\beta$ can be generated by a randomized leaky-bucket adversary of injection rate $\rho$ and burstiness~$\beta$ with a positive probability. 
The advantage of a randomized leaky-bucket adversary is that it can be simulated in a natural manner.
Section~\ref{sec:broadcast-algorithms} reviews broadcast algorithms discussed in detail later on.
These algorithms are known in the literature, and we provide brief descriptions for completeness sake.
Section~\ref{sec:worst-case-bounds} presents new worst-case upper bounds on the latency of adversarial broadcasting.
They are for two specific cases.
One is for algorithms designed for a perpetual channel with named stations when leaky bucket adversaries are restricted to have individual injection rates.
We investigate the feature of algorithms known as ``old-go-first'' in this context. 
The other is for algorithms designed for ad-hoc channels when they are executed on perpetual channels.
We investigate the range of injection rates for which latency is bounded, and in particular strong universality of broadcast algorithms.
Section~\ref{sec:simulations-randomized-adversaries} proposes a methodology of experiments of adversarial packet injections and presents outcomes of specific experiments.
Section~\ref{sec:conclusion} concludes with final remarks.

An earlier version of this work was presented at the $18$th IEEE International Symposium on Network Computing and Applications (NCA 2019)~\cite{AldawsariCK19}.

\section{Technical Preliminaries}

\label{sec:technical-preliminaries}

A multiple-access channel is a broadcast network that makes it possible for  each user to transmit a message directly to all the other users.
A multiple-access channel consists of a shared communication medium and communication agents using the shared medium; the agents are called \emph{nodes} or \emph{stations}. 
We consider a synchronous variant of such networks, in which an execution of a communication algorithm is structured as a series of consecutive rounds.

Shared channels can be categorized depending on whether the set of nodes is set permanently or it may change in time. 
A \emph{perpetual channel} has a permanent set of some $n$  nodes, which have fixed names to identify them; a node's \emph{name} is a unique integer in the range $[0,n-1]$.
In an \emph{ad-hoc channel}, there is a potentially unbounded supply of nodes, which do not have any individual identifiers.
Nodes may join and leave an ad-hoc channel, while only finitely many nodes participate actively in an execution of a communication algorithm at any round.

Packets obtained by the stations are encapsulated in messages, which then are  broadcast on the channel.
A transmitted message contains at most one packet.
The duration of a round and the size of a message are mutually scaled such that it takes precisely a whole round to transmit a message.
This means that transmissions of messages by different stations overlap if and only if they occur in the same round.

\Paragraph{Sharing a channel.}

Each station obtains feedback from the channel in every round.
The feedback in a round is the same for each station, and does not depend on whether the station transmits a message at this round or not.
In particular, a transmitting node and a pausing node receive the same feedback.

A node \emph{hears} a message when it receives it successfully.
It is the critical property of multiple-access channels how the multiplicity of concurrent transmissions affects hearing the transmitted messages.
There are the following three relevant cases of the multiplicities of transmissions in a round: no nodes transmit, exactly one node transmits, and multiple (more than one) nodes transmit.
When multiple nodes transmit in a round then this is referred to as a \emph{collision} occurring in this round.
These three cases are discussed next.

If no node transmits in a round then the channel is \emph{silent}, which is reflected by the corresponding \emph{silence} feedback received from the channel by every node.
If only one node transmits, then all the nodes can hear the transmitted message as  the feedback they receive from the channel, including  the transmitting node.
If multiple nodes transmit in the same round, which creates a collision, then the effect is such that the messages interfere with one another and none can be heard by any attached node.
If a collision occurs in a round then every node obtains a \emph{collision signal} as the feedback from the channel.
A round in which no message is heard on the channel is called \emph{void}; a void round is either silent or created by a collision.
Multiple-access channels come in two variants, which are determined by the feedback coming from a channel when a message is not heard.
A channel \emph{without collision detection} has the property that both silence and collision signals are identical.
In a channel \emph{with collision detection} these respective signals are different, so that the nodes can distinguish between the two cases of silence and collision.

\Paragraph{Dynamic broadcasting.}

Stations participate in executing a distributed broadcast algorithm.
The goal is to perform \emph{dynamic broadcasting}, in the sense that packets get generated and then injected into the nodes continually, while  the nodes strive to have these packets heard on the channel.
All the stations begin executing a given algorithm together starting from the  first round.

Dynamic packet generation  leads to additional categorization of shared channels.
In the \emph{queue-free} model, each generated packet gets injected into a new station that has been previously passive.
Such an injection  makes the station active for the duration of handling the packet; this model was considered for example in Bender et al.~\cite{BenderFHKL-SPAA05}.
In the \emph{queuing} model, each station can handle multiple packets at the same time.
The packets handled by a station get stored in a private buffer space referred to as this station's \emph{queue}; this model was considered for example in H{\aa }stad et al.\cite{HastadLR-SICOMP96}.
In this paper we consider the queueing model.  
Every station has a potentially unbounded buffer to store queued packets.
A packet queued by a station gets dequeued immediately after it has been heard on the channel.
Packets do not get dropped by stations without being heard on the channel, neither due to timeouts nor for any other reason. 

Let the nodes that have packets to transmit in the beginning of a round be called \emph{active} in this round and \emph{passive} otherwise.
The event of injecting a packet into a passive station results in an \emph{activation} of this station.
If a positive integer~$k$ is an upper bound on the number of stations that can be activated in a round, then the communication channel is called \emph{$k$-activating}.

In this paper, we consider a perpetual channel with some number~$n$ of stations and the adversaries constrained by the $1$-activation constraint.
The motivation for having at most one node activated in a round comes from the real-world applications, when new packets  typically get injected into nodes that are already active, and rarely a passive node will get activated by obtaining packets to broadcast.
This latter interpretation means also that having multiple passive nodes activated in a round occurs so rarely  that it can be disregarded without distorting the performance of broadcasting, as modeled in simulations.

Historically, perpetual channels without constraints on the number of activations in a round were typically studied.
Having a constant number of named stations allows to develop deterministic broadcast algorithms with bounded latency for all injection rates $\rho<1$ and stable algorithms for injection rate $\rho = 1$, see the work by Chlebus et al.~\cite{ChlebusKR-TALG12,ChlebusKR-DC09} and Anantharamu et al.~\cite{AnantharamuCKR-JCSS19}.
Restricting packet injections to $1$-activating patterns allows to develop deterministic broadcast algorithms for dynamic channels that are stable for sufficiently small injection rates, as was shown in Anantharamu et al.~\cite{AnantharamuC15}.

\Paragraph{Categories of broadcast algorithms.}

A broadcast algorithm is \emph{full sensing} when all the nodes listen to the channel at all times. 
We understand ``listening to the channel'' as undergoing state transitions determined by the feedback from the channel, which is opposed to ``ignoring the feedback'' by idling in an initial state.
An algorithm is \emph{activation based} when a passive station  ignores the feedback from the channel and starts participating in a round when it gets activated by having a packet injected into it.
An active station executing an activation-based algorithm listens to the channel when it has packets to transmit, but it stops listening to the channel as soon as it has no pending packets.
An activation-based algorithm is called \emph{acknowledgement based} when a station resets its state to an initial one after a successful transmission of a message.

A message transmitted on the channel includes at most one packet but it may consist of only control bits.
If an algorithm does not use control bits at all, then messages transmitted in  the course of its execution contain only packets; such an algorithm is called a \emph{plain packet} one.
A node executing a plain-packet algorithm cannot transmit a message at all if it does not have a pending packet in its queue.
A general algorithm that allows stations to transmit messages with control bits is also called \emph{adaptive}.
A node executing an adaptive algorithm may transmit a message consisting of only control bits.

The actions performed in a round by a node executing a full-sensing algorithm or an activation-based one by an active station are as follows.
The node first either transmits a message or pauses, as determined by its state.
Then the nodes obtains feedback from the channel, which is the same at all the stations.
Next, the node may have a number of packets injected into it at this round.
Finally, the node performs local computation, which can be interpreted as a state transition.
Local computation involves the following actions.
If packets were injected then they are enqueued in the node's queue.
The private variables get updated, depending on what occurred in the round up to this moment.
The node decides if to transmit at the next round or not, and if so then it builds a message to be transmitted.
In particular, if the queue includes packets then the message to be transmitted may include a packet.

\Paragraph{Performance metrics.}

The number of packets in a station's queue in a round is this stations's \emph{queue size} at the round.
If a packet is injected into a station in a round~$t_1$ and is heard on the channel in a round $t_2>t_1$ then $t_1-t_1$ is the number of rounds it spends in the queue, which is this packet's \emph{delay}.

The performance of broadcast algorithms is assessed with respect to two performance metrics.
\emph{Latency}  is the maximum packet delay in an execution.
The \emph{number of queued packets} or simply the \emph{queues} is the maximum sum of the queue sizes of the stations in a round in an execution.
A broadcast algorithm is \emph{stable} in an execution if there is an upper bound on the number of queued packets.

\section{Adversarial Packet Injection}

\label{sec:adversarial-packet-injection}

The  quantitative restrictions on packet generation and injections are formulated as adversarial models.
The adversaries we consider are all specializations of the leaky-bucket concept.
The behavior of an adversary can be understood as an \emph{execution}, which proceeds through consecutive rounds.
At each round of such an execution, the adversary first generates packets and next injects these packets into stations.
\emph{Generation} means producing a finite number of packets in a round.
\emph{Injection} denotes enqueuing each generated packet into a queue at  some station.

Adversarial models provide a framework that allows to study the worst-case performance of broadcasting, which is its greatest appeal.
We can build randomness into an adversarial model to capture average performance as well.
An adversary that has either the process of generation of packets or their injection determined randomly is called \emph{randomized} and otherwise it is \emph{worst case}.
The leaky-bucket adversarial paradigm will be considered for both the worst-case and randomized-case senses. 
An adversary may have an upper bound on the number of stations activated in a round as an additional parameter, in particular, it could be $k$-activating, for an integer $k>0$.
Similarly, an adversary may be additionally constrained by  upper bounds on the frequency of packets injected into each individual station.

\Paragraph{Leaky-bucket adversaries.}

A leaky-bucket adversary is determined by a pair of numbers.
One is the \emph{injection rate}, denoted $\rho$, which is a real number satisfying $0<\rho\le 1$.
The other is the \emph{burstiness component}, denoted~$\beta$, which is a real number satisfying $\beta\ge 1$.
Together they make the \emph{type $(\rho,\beta)$} of the adversary.
For a contiguous time interval~$\tau$ of $|\tau|$ rounds, the adversary may generate up to $\rho\cdot |\tau| + \beta$ packets during the rounds in~$\tau$.

The definition of an adversary of the type $(\rho,\beta)$ constraints both average numbers of created packets in large intervals and also bursts of numbers of packets generated in short intervals. 
An injection rate~$\rho$ can be interpreted as the average number of generated packets,  if  averaging over large intervals.
Indeed, the average number created in intervals of $t$ rounds is at most the following: $\frac{\rho \cdot t+\beta}{t}\rightarrow \rho$, where  $t\rightarrow\infty$.
The adversary may generate at least one packet in a round and there is an upper bound on the number of packets created in one round.
The maximum number of packets that can be generated in one round is the \emph{burstiness} of the adversary.
Indeed, let $\tau$ be a time interval of just one round: $|\tau|=1$.
Then at most $\rho\cdot |\tau| + \beta = \rho+\beta $ new packets can be generated in $\tau$, and the following inequalities hold: $1< \rho+\beta\le 1+\beta$.
A generated packet is immediately injected into a station by enqueuing it in the station's queue.
The adversary also determines where to inject each newly generated packet, which may be constrained though a further stipulation of the adversarial model.

\Paragraph{The four step bucket process.}

Next, we describe a process of controlling the number of packets injected in a round by a leaky-bucket adversary of a given type  through the concept of a bucket.
We refer to the following process as  \emph{four-step bucket process of type $(\rho,\beta)$}.

The four-step process proceeds as an execution through a sequence of rounds.
The actions in a round are determined as follows.
Let $D$ be a variable, which we call a \emph{bucket}.
The bucket variable~$D$ is initialized to~$\beta$ at round zero.
The number of packets generated in each of the following rounds is determined in  four steps:
\begin{enumerate}
\item
First, reset the bucket to $D\leftarrow \min[D+\rho,\beta]$.
\item
Then, choose a non-negative integer $X$. 

\item
Next, set the number of actually generated packets to  $j\leftarrow \min\{\lfloor D\rfloor,X\}$.
\item
Finally, update the bucket variable to $D\leftarrow D-j$.
\end{enumerate}
The first step represents the bucket \emph{leaking} at rate~$\rho$, which is the rate with which the bucket's available capacity grows.
The second step involves a proposed quantity $X$ of packets to be generated, assuming this quantity is consistent with the adversarial model.
The third step represents a verification of the  available capacity of the bucket $\lfloor D\rfloor$, as determined by the recent generations of packets and continuous leaking.
The fourth step updates the available capacity of the bucket to account for the number~$j$ of packets created in this round.

\begin{proposition}
\label{pro:one}

Consider a sequence of packets generated at each round by an adversary in an unbounded execution.
The sequence is consistent with a leaky-bucket adversarial type $(\rho,\beta)$ if and only if it is consistent with the four-step bucket process  of type $(\rho,\beta)$. 
\end{proposition}

\begin{proof}
Consider the four-step bucket process.
We may observe that the invariant $0\le D\le \beta$ holds after each round.
We show this by induction on the round numbers.
The base case of induction follows from the initialization of $D$ to~$\beta$.

For the inductive step, assume that the invariant holds in a round and consider the next round.
The bucket will satisfy $D\le \beta$ in the next round because the first step of generation is the only one when $D$ could be increased.
The bucket will satisfy $D\ge 0$ in the next round because its decrease in the fourth step is preceded by a verification in the third step. 

Now, continuing the inductive step, we are ready to show that if a sequence of numbers of packets generated at each round is consistent with the four-step bucket manipulation of type $(\rho,\beta)$ then it is consistent with a leaky-bucket adversary of type $(\rho,\beta)$.
Let $D$ be the bucket variable and let $\tau=[t_1,t_2]$ be a time interval.
The variable $D$ is at most $\beta$ when round~$t_1$ begins, by the invariant. The bucket~$D$ is incremented by at most~$\rho$ in the beginning of each round in~$\tau$, for the total of at most $\rho \cdot |\tau|$ during all the rounds in~$\tau$.
This can be at most balanced by injecting packets, suitably constrained by updating $D$ in the fourth step.
Generating each packet results in decrementing the variable~$D$ to $D-1$ in the fourth step.
Such a decrement is always possible to perform, because during the third step, it is verified that the number of packets generated in a round does not surpass~$D$.
It follows that the adversary can inject at most $(t_2-t_1)\rho +\beta$ packets during~$\tau$.

Next, we show that if a sequence of numbers of packets generated at each round is consistent with a leaky-bucket adversary of type $(\rho,\beta)$ then it is consistent with the four-step bucket manipulation  of type $(\rho,\beta)$. 
Let $\tau=[0,t]$ be a time interval, for $t>0$.
The bucket variable is initialized to~$\beta$ at round~$0$, which allows to inject $\lfloor \beta\rfloor$ packets in the first round.
The bucket variable is an upper bound on the number of packets that can be injected at each round, per the third step, and is decreased only to record actual packet generation in the fourth step.
So the total number of packets generated in $\tau$ is accounted for by all the decrements of~$D$.
\end{proof}

\Paragraph{Randomized adversaries.}

We define \emph{randomized leaky-bucket adversaries} of type~$(\rho,\beta)$, where $0< \rho \le 1$ and $\beta\ge 1$.
The adversarial model is determined through the four-step bucket manipulation process.

Any sequence of numbers of generated packets is consistent with the definition of a leaky-bucket adversary of type~$(\rho,\beta)$, by Proposition~\ref{pro:one}.
The adversary is defined by how the number of packets~$X$ in the third step of packet generation is determined.
We treat $X$ as a random variable in this randomized adversarial model.
The random variable~$X$ has a Poisson distribution with parameter $\lambda>0$ if $\Pr (X=i) =e^{-\lambda}\cdot \frac{\lambda^i}{i!}$, for each integer $i\ge 0$; see Mitzenmacher and Upfal~\cite{MitzenmacherUpfal-book17}.

The adversary uses a random variable~$X$ that is a Poisson distribution with parameter~$\lambda$ equal to the injection rate~$\rho$ in type $(\rho,\beta)$.
The average number of generated packets through an execution of actions of a randomized adversary of type $(\rho,\beta)$ is less than $\rho$.
This is because the expectation of the random variable $X$ is the second step of generation is exactly~$\rho$, but the  number of generated packets is restrained by the bucket's capacity, as reflected in the third step.


\begin{proposition}
\label{pro:simulating-adversary}

Each sequence of numbers of generated packets during a time segment $[t_1,t_2]$, for $0\le t_1<t_2$, that is consistent with the general adversary of type $(\rho,\beta)$, can be generated with a positive probability by the randomized adversary of type $(\rho,\beta)$. 
\end{proposition}

\begin{proof}
Proposition~\ref{pro:one} states that consistence with an adversary of a type $(\rho,\beta)$ is equivalent to the four-step bucket process  of type $(\rho,\beta)$, so we can consider the corresponding bucket process.
For each value~$D$ such that $0\le D\le \beta$, and for each integer $j$ such that $0\le j\le D$, the probability of generating $j$ packets in a given round is positive, as determined by the second step of the generation.
\end{proof}

\Paragraph{Individual injection rates.}

The adversaries with individual injection rates are defined for a perpetual channel with a given number~$n$ of stations attached to it.
The stations are identified by their names in the interval $[0,n-1]$.
The adversarial model is a specialization of the general leaky-bucket adversaries,  so the adversary is of some leaky-bucket adversary's type $(\rho,\beta)$.

On top of the general leaky-bucket constraint, there are additional restrictions on injecting packets into individual stations.
Namely, each station $i$ has its \emph{individual injection rate~$\rho_i$} assigned to it, such that $0\le \rho_i\le 1$.
All these injection rates satisfy $\sum_{i=0}^{n-1} \rho_i = \rho$.
For each station $i$ and a contiguous time interval~$\tau$ of $|\tau|$ rounds, the adversary can inject at most $\rho_i\cdot |\tau| + \beta$ packets into the station~$i$ during the rounds in~$\tau$.
These packets are chosen from among all the packets that can be generated during~$\tau$ contiguous rounds according to the general leaky-bucket constraint; there are at most $\rho\cdot |\tau| + \beta$ such packets.

\section{Review of Broadcast Algorithms}

\label{sec:broadcast-algorithms}

There are two general paradigms to structure  deterministic broadcast in a distributed manner.
Algorithms designed specifically for a perpetual multiple-access channel with a fixed set of nodes attached to the channel, each with a unique name, could  operate by having the nodes exchange a token.
The token visiting a node allows the node to transmit, which prevents collisions.
Such algorithms could be called \emph{token} ones; see Chlebus et al.~\cite{ChlebusKR-TALG12,ChlebusKR-DC09} and Anantharamu et al.~\cite{AnantharamuCKR-JCSS19}.
The other paradigm is for ad-hoc channels, where nodes are dynamically generated.
The idea is to assign temporary implicit names: this works for a setting in which at most one new node is added to the system in a round.
Such algorithms could be called \emph{ad-hoc} ones; see Anantharamu et al.~\cite{AnantharamuC15}.

A broadcast algorithm is \emph{universal} if it is stable for all injection rates  less than~$1$,  and it is \emph{strongly universal} if it is universal and there exists a range $[0,c]$, for $0<c<1$, such that if an injection rate $\rho$ satisfies $0<\rho<c$ then the number of packets in the queues is a function of only the type of the adversary $(\rho,\beta)$, rather than the number of stations~$n$ in a perpetual channel.
A broadcast algorithm has \emph{universal latency} if it has bounded  latency for each injection rate smaller than~$1$.
A broadcast algorithm  has \emph{strongly universal latency} if it has universal latency and there exists a range $[0,c]$, for $0<c<1$, such that if injection rate $\rho$ satisfies $0<\rho<c$ then packet delay is a function of only the type of the adversary~$(\rho,\beta)$.
The \emph{throughput}  is the maximum injection rate for which an algorithm is stable, if such an injection rate exists.

Broadcasting by deterministic algorithms can be more effective on perpetual multiple-access channels than on channels executing ad-hoc algorithms, due to the stabilizing effect of a fixed set of stations.
A perpetual channel can achieve throughput~$1$, as shown in Chlebus et al.~\cite{ChlebusKR-DC09}.
There are algorithms for perpetual channels with universal packet latency, as shown in Anantharamu et al.~\cite{AnantharamuCKR-JCSS19}. 
In contrast to this, ad-hoc algorithms cannot handle injection rates that are at least $\frac{3}{4}$ with bounded  latency, as shown in Anantharamu and Chlebus~\cite{AnantharamuC15}, which implies that there are no universal algorithms for ad-hoc channels.

We consider three groups of algorithms: deterministic token ones, deterministic ad-hoc ones, and randomized ones.
In what follows, we give an overview of  these algorithms.

\Paragraph{Algorithms for perpetual channels.}

Algorithm \textsc{Round-Robin-Withholding} (\textsc{RRW}) is a plain-packet  full-sensing algorithm for channels without collision detection.
Algorithm \textsc{Search-Round-Robin} (\textsc{SRR}) is a  plain-packet full-sensing algorithm  for channels with collision detection.
In each of these algorithms, the stations share a virtual token, which is circulated among the stations, with names in the interval $[0,n-1]$ ordered into a cycle, in a round-robin manner.
A station with the token withholds the channel to transmit all its packets one by one.
A silent round is interpreted as a signal for a next station to take over.
The station that takes over is determined differently in each of the algorithms.
In algorithm RRW, the station that takes over transmitting is the next station  in the circular order after the one  holding the token.
In algorithm SRR, the station that takes over transmitting is the next station  in the circular order after the one  holding the token that has a packet pending transmission; such a station is determined by binary search using collision detection. 
The algorithms RRW and SRR were proposed in Chlebus et al.~\cite{ChlebusKR-TALG12} and shown to have universal latency.

Each of the algorithms \textsc{RRW} and \textsc{SRR} can be modified by using the approach called ``old-go-first,'' see Anantharamu et al.~\cite{AnantharamuCKR-JCSS19}.
This approach works as follows. 
A \emph{phase} consists of a segment of rounds in which the token makes a full cycle and returns to the starting point.
The packets that are injected in a phase are considered ``new'' during the phase and become ``old'' when the next phase starts.
In the course of a phase, the new packets are ignored and only the old ones are broadcast.
The algorithms modified this way are called  \textsc{Old-First-Round-Robin-Withholding} (\textsc{OF-RRW}) and \textsc{Old-First-Search-Round-Robin} (\textsc{OF-SRR}), respectively.

Algorithm \textsc{Move-Big-To-Front} (\textsc{MBTF}) is an adaptive full-sensing algorithm for channels without collision detection.
The stations have their names ordered into a virtual list, beginning with the ordering determined by their names. 
A station is considered big if it has at least $n$ packets in its queue.
There is a virtual token that traverses the list in a round-robin manner.
If a big station receives the token then it is moved to the beginning fo the list and continues transmitting as long as it is big.
If a station that is not big receives a token then it transmits only one packet, assuming it has any, and otherwise pauses in the round.
Algorithm MBTF was proposed in Chlebus et al.~\cite{ChlebusKR-DC09} and shown to attain throughput~$1$ and have universal latency.

\Paragraph{Algorithms for $1$-activating ad-hoc channels.}

We present briefly deterministic distributed algorithms designed for $1$-activating ad-hoc channels.
The ad-hoc model and the broadcast algorithm for this model were given in Anantharamu and Chlebus~\cite{AnantharamuC15}.

Algorithm \textsc{Counting-Backoff}  is a plain-packet activation-based algorithm for ad-hoc channels with collision detection.
The active nodes are stored on a virtual distributed stack.
The stack is implemented such that an active station remembers its distance from the top of the stack; the station on top of the stack is of distance zero for the top.
A station that is on top of the stack keeps transmitting packets as long as its queue is nonempty.
A station activated in a round considers itself to be pushed on the stack by the activation and so transmits in the next round.
If a station transmits in a round that produces a collision then the station considers itself to be second from the top of the stack, since the collision is with a newly activated station that has just got pushed on the stack.
In general, a collision makes all the stations on the stack increase their distance to the top of the stack by one, and a silent round makes all the stations on the stack decrease their distance to the top of the stack by one.

In an execution of algorithm \textsc{Quadruple-Round}, time is partitioned into segments of length four, which are considered consecutively one by one.
For a processed time segment, the stations activated in one of the four rounds in the segment have an opportunity to broadcast their packets.
The existence of such stations is verified in a binary search manner, with the four rounds associated with leaves.
In the first round of searching, each such a station is to broadcast its packet.
If this first round results in silence, then this means that no station was activated in the given time segment and the algorithm advances to the next segment; otherwise, if collision is detected, then the binary search branches off to consider the left and right subtrees.
If a packet is heard, then the transmitting station does not withhold the channel, instead, it will have an opportunity to transmit again in the next iteration of the binary search.

The active stations executing algorithm \textsc{Queue-Backoff} are organized as a distributed first-in-first-out queue, which they join in the order of their activation times. 
A station at the front of the queue transmits its packets, and the last packet has an ``over'' bit attached, used to prompt the next station to start its transmissions.
Additionally, each message carrying a packet includes control bits describing the current size of the queue.
A newly activated station transmits immediately, which results in a collision, unless the queue is empty.
If a collision occurs, the front station learns that the queue has just increased, and this count is reflected in the messages the station transmits.
A station that joins the queue identifies its position by the size of the queue when it learns it for the first time, from which it subtracts the number of immediately preceding collisions created by other stations joining the queue.

These three algorithms for ad-hoc channels presented above do not use the names of nodes, even if they are available. 
They have ranges of injection rates with bounded latency  even when there is no fixed set of nodes attached to the channel and the adversary has the power to ``create'' new stations by injecting packets into them, but at most one new station in a round; see Anantharamu and Chlebus~\cite{AnantharamuC15}.

\Paragraph{Randomized backoff algorithms.}

Now we discuss two backoff algorithms, see H{\aa }stad et al.~\cite{HastadLR-SICOMP96} for a discussion of backoff algorithms.
These are the \textsc{Binary-Exponential-Backoff} (\textsc{BEB}) and \textsc{Quadratic-Backoff} (\textsc{QB}).
Each of them is randomized and acknowledgement-based.
As soon as a packet is made heard on the channel, the next available packet is processed immediately. 
A node processes its packets in the order of injection.
The backoff algorithms are considered in their windowed versions. 
A \emph{window} is a contiguous segment of rounds starting from the round of a successful transmission or a collision.
The lengths of windows may vary; the window just after a successful transmission is of size~$1$.
When a new packet is available for processing, the node selects a round uniformly at random in the current window, then waits for this round to occur and transmits the packet.
Since the first window consists of one round, a new packet is  transmitted immediately.

The backoff algorithms differ among themselves in how the sizes of consecutive windows  increase.
For \textsc{Binary-Exponential-Backoff}, the $i$th window size is determined as~$2^{i}$, and  for \textsc{Quadratic-Backoff}, the $i$th window size is defined to be~$i^2$.

The two randomized algorithms can be considered with an upper bound on the window size, as the binary exponential backoff is used in the implementation of the Ethernet.
For instance,  the $i$th window size for \textsc{Binary-Exponential-Backoff} could be $2^{\min(10, i) }$, which is exactly as in the Ethernet, and for \textsc{Quadratic-Backoff}, the $i$th window size could be $(\min(i,32))^2$.
Observe that, with these two choices of constants, the maximum size of a window is the same in  \textsc{BEB} and~\textsc{QB}, since $2^5=32$.

\Paragraph{Knowledge of algorithms.}

A communication system we consider consists of a multiple access channel and an adversary generating and injecting packets.
We say that an algorithm \emph{knows} a parameter of such a communication system if this parameter can be used as part of code.
The algorithms we consider never know the adversary controlling packet injection.
The algorithms designed for perpetual channel are token based, and manipulating a token requires the knowledge of the number of stations~$n$.
The algorithms for ad-hoc channels do not know the number of stations by design, since ad-hoc channels do not have a fixed set of stations; the lack of this knowledge extends to the case when such algorithms are executed on a perpetual channel with a fixed set of stations.

\section{Upper Bounds on Queues and Latency}

\label{sec:worst-case-bounds}

We investigate the worst-case upper bounds on the latency of adversarial broadcasting in two special cases.
One is for algorithms designed for a perpetual channel with named stations when leaky bucket adversaries are restricted to have individual injection rates.
The effect of individual injection rates, as compared to the general leaky-bucket adversarial model, is that the old-go-first versions of algorithms have greater bounds on packet latency rather than smaller, when compared to the regular versions.
The other is for algorithms designed for ad-hoc channels when they are executed on perpetual channels.

Deterministic algorithms for $1$-activating ad-hoc channels were given in Anantharamu and Chlebus~\cite{AnantharamuC15} that can handle injection rates up to $\frac{1}{2}$ in a stable manner.
A fixed set of nodes of a perpetual channel may  have a stabilizing effect on these algorithms.
Indeed, we show that a particular algorithm designed for ad-hoc channels has  \emph{strong} universal latency when executed on $1$-activating perpetual channels.

\subsection{Old-Go-First with individual injection rates}

\label{sec:individual}

We begin with token algorithms designed for perpetual channels.
It was shown in Anantharamu et al.~\cite{AnantharamuCKR-JCSS19} that algorithm \textsc{OF-RRW}  executed  by $n$ stations  against the adversary of type~$(\rho,\beta)$ has the following asymptotically-tight performance bounds: the number of packets simultaneously queued in the stations is at most $\frac{2\rho }{1-\rho}\cdot n +\beta$ and the packet latency is at most $\frac{2}{1-\rho}\cdot n +\beta(1+\rho)$.
In contrast to that, algorithm \textsc{RRW} has he following asymptotically-tight performance bounds: the number of packets simultaneously queued in the stations is at most $\frac{2\rho }{1-\rho}\cdot n +\beta$ and the packet latency is at most $\frac{2-\rho}{(1-\rho)^2} \cdot n+ \frac{\beta}{1-\rho}$.
It follows that algorithm  \textsc{OF-RRW} has superior worst-case performance bounds as compared to algorithm  \textsc{RRW}, better by a factor $1/(1-\rho)$, which grows unbounded as the injection rate~$\rho$ increases towards~$1$.

Next, we show that if an adversary is restricted by individual injection rates then this phenomenon does not hold.


\begin{theorem}
\label{thm:round-rw}
Algorithm  \textsc{RRW} has at most $\frac{\rho  }{1-\rho}\,n+\beta$ packets queued and its latency is at most $\frac{2-\rho}{1-\rho}\,n+ \beta$ when executed on a perpetual channel with $n$ stations against the $(\rho,\beta)$ adversary with individual injection rates.
\end{theorem}

\begin{proof}
We consider an execution in which the adversary injects at full power, with the effect of burstiness considered separately.
As the phases pass, their length increases approaching a limit, with possible fluctuations.
If the adversary is injecting at full power into a particular station, then the numbers of rounds between two consecutive injections differ by at most~$1$, due to rounding.
If this occurs for all the stations in the same phase then this phase may be extended by $n$ rounds.
Burstiness can make queues increase by $\beta$ in one phase and contribute $\beta$ to the length of the phase. 
Let $p$ denote the number of packets in a phase in an equilibrium state, while disregarding burstiness: the number of packets at the start of a phase is the same as after the phase is over.
Such a phase takes $n+p$ rounds, and the number of packets injected during the phase is $\rho(n+p)$.
This gives the equation $p=\rho(n+p)$,  which can be solved for $p$ to give $p=\frac{\rho }{1-\rho}\, n$.
The number of packets in a phase can be increased by injecting up to $\beta$ packets extra packets, for a total of $\frac{\rho }{1-\rho}\, n+\beta$.

A phase consists of the rounds spent on transmitting packets, then $n$ silent rounds that result in the token advancing through the stations, and possibly $n$ extra rounds due to individual fluctuations of the number of packets injected individually in each station.
This bound is therefore as follows: $2n +\frac{\rho }{1-\rho}\, n+\beta$,
which can be simplified to the claimed form by algebra.
\end{proof}

Next, we discuss algorithms SRR and OF-SRR.
Algorithm \textsc{OF-SRR}  executed  by $n$ stations  against an adversary of type~$(\rho,\beta)$ has the following asymptotically-tight performance bounds: the number of packets simultaneously queued in the stations is at most $\frac{4\rho }{1-\rho}\cdot n  +\beta$ and packet latency is at most $\frac{4}{1-\rho}\cdot n +\beta(1+\rho)$.
In contrast to that, algorithm \textsc{SRR} has he following asymptotically-tight performance bounds: the number of packets simultaneously queued in the stations is at most $\frac{4\rho }{1-\rho}\cdot n +\beta$ and packet latency is at most $\frac{4-2\rho}{(1-\rho)^2}\cdot n +\frac{\beta}{1-\rho}$; see Anantharamu et al.~\cite{AnantharamuCKR-JCSS19}.

Algorithms \textsc{SRR} and \textsc{OF-SRR} have greater upper bounds on packet latency than \textsc{RRW} and \textsc{OF-RRW}, respectively.
Algorithms \textsc{SRR} and \textsc{OF-SRR} have a property that their packet latency becomes $\cO(\log n)$ for suitably small injection rates that are $\cO(1/\log n)$, see Anantharamu et al.~\cite{AnantharamuCKR-JCSS19}.
Algorithm  \textsc{OF-SRR} has the superior worst-case performance bounds  compared to algorithm  \textsc{SRR} by a factor $1/(1-\rho)$, which grows unbounded  if $\rho$ converges to~$1$.
If an adversary is restricted by individual injection rates, then this phenomenon does not occur, as stated next.


\begin{theorem}
\label{thm:search-rw}
Algorithm  \textsc{SRR} has at most $\frac{2\rho  }{1-\rho}\,n+\beta$ packets queued in a round and its latency is at most $\frac{3-\rho}{1-\rho}\,n+ \beta$ if executed on a perpetual channel with $n$ stations against the $(\rho,\beta)$ adversary with individual injection rates.
\end{theorem}

\begin{proof}
We denote by $p$ the number of packets in a phase in equilibrium state, while disregarding burstiness: the number of packets at the start of a phase is the same as after the phase is over.
There may be up to $2n-1$ void rounds in a phase, so a phase takes at most $2n+p$ rounds, while the number of packets injected during the phase is $\rho(2n+p)$.
This gives the equation $p=\rho(2n+p)$, which solved for $p$ gives $p=\frac{2\rho }{1-\rho}\, n$.
The number of packets in a phase can be increased by injecting up to~$\beta$ packets extra packets, for a total of $\frac{2\rho }{1-\rho}\, n+\beta$.

A phase consists of the rounds spent on transmitting packets, then $2n-1$ silent rounds that result in the token advancing through the stations, and possibly $n$ extra rounds due to individual fluctuations of the number of packets injected individually in each station.
This bound is therefore as follows: $3n +\frac{2\rho }{1-\rho}\, n+\beta$, which can be simplified to the claimed form by algebra.
\end{proof}

\subsection{Ad-hoc algorithms in perpetual channels}

We will consider algorithms designed for ad-hoc channel when executed on a perpetual channel with a fixed set of stations against $1$-activating adversaries.

Algorithm \textsc{Counting-Backoff} has packet latency $\frac{3\beta-3}{1-3\rho}$ in ad-hoc $1$-activating channels, for injection rates $\rho$ that are less than~$\frac{1}{3}$; see Anantharamu et al.~\cite{AnantharamuC15}.
This algorithm is not stable for injection rates at least $\frac{1}{3}$, since the stack may stay forever nonempty in an execution, and the station at the bottom is starved for access to the channel.
This property holds in both ad-hoc and perpetual channels.

Algorithm \textsc{Queue-Backoff} is an adaptive activation-based algorithm for ad-hoc  channels without collision detection that has packet latency $4\beta-4$ against $1$-activating adversaries with injection rate~$\frac{1}{2}$, see Anantharamu and Chlebus~\cite{AnantharamuC15}.
Next we show that the algorithm has universal packet latency on perpetual channels against $1$-activating adversaries.
It follows that algorithm \textsc{Queue-Backoff} has strong universal latency in perpetual channels  against $1$-activating adversaries.


\begin{theorem}
Algorithm \textsc{Queue-Backoff} is an adaptive algorithm with strongly-universal packet latency  if  executed on a $1$-activating perpetual channel with collision detection.
It has queues at most $\frac{\rho}{1-\rho}\cdot n+\beta $ and latency at most $\frac{\rho}{(1-\rho)^2}\cdot n +\frac{\beta}{1-\rho} $  against $(\rho,\beta)$  adversaries  with $\rho<1$ an $n$ stations attached to such a channel.
\end{theorem}

\begin{proof}
We investigate how  the adversary can organize an execution to maximize the number of packets and the delay of some packet.
A specific execution will consist of two parts.
In the first part, the adversary will work to make all the active stations store as many packets in total as possible.
This is accomplished by activating a station whenever there is an opportunity to do this and continuously injecting into this station at full power until a next station is activated.
After the goal of the first part is accomplished, a station is activated with a new dedicated packet that we want to delay as much as possible.
This is accomplished by injecting at full power only into stations that are in front of the dedicated packet's station in the distributed queue.

The first part begins by the queue growing to encompass all the available $n$ stations.
This is feasible because injecting a packet at every other round is sufficient to maintain the current size of the queue, in the sense that immediately afer a station leaves the queue at least one new station is activated.
We define a phase to be a contiguous segment of rounds during which  $n$ stations become active. 
A phase includes the $n$ collision rounds and the rest consists of the rounds with packets heard on the channel. 

Let $p$ be the maximum number of packets stored in all the active stations while disregarding the effect of burstiness, so that the ultimate maximum number of packets $m$ is $m=p+\beta$.
As the first part of the execution progresses, the number of packets approaches the solution of the equation $p=\rho (n+p)$.
This equation can be solved for $p$ to yield $m=\frac{\rho}{1-\rho}\cdot n+\beta$.

Once the first part ends, a packet is injected into a newly activated station with the goal to delay its successful transmission by injecting at full power in the stations preceding the packet's one.
This results in the following number of rounds spent by the packet waiting:
\[
m(1+\rho + \rho^2+\cdots) =\frac{m}{1-\rho}
\ , 
\]
which is equivalent to the claimed form by algebra.
\end{proof}

Algorithm \textsc{Queue-Backoff} is adaptive, in that it uses control bits in messages. 
We conjecture that this is a necessary property of algorithms with strongly universal latency.

\begin{conjecture}
There is no plain-packet algorithm that has strongly universal latency when executed on $1$-activating perpetual channels. 
\end{conjecture}

Algorithm \textsc{Quadruple-Round} is a full-sensing plain-packet algorithm for ad-hoc  $1$-activating  channels with collision detection that has packet latency $2\beta+4$ and queues $\beta+\cO(1)$ against $(\frac{3}{8},\beta)$ adversaries, see Anantharamu and Chlebus~\cite{AnantharamuC15}.
These bounds hold a fortiori for perpetual channels.

Next, we show that the algorithm attains bounded latency for a larger range of injection rates   in perpetual channels compared to ad-hoc ones.


\begin{theorem}
Algorithm \textsc{Quadruple-Round}  has queues at most $ \frac{\rho}{3-7\rho}\cdot n+\beta$ and latency at most $ \frac{7\rho}{(3-7\rho)^2} \cdot n +\frac{n+7\beta}{3-7\rho}$, if executed on a $1$-activating perpetual channel with collision detection  against $(\rho,\beta)$  adversaries with $\rho<\frac{3}{7}$ and $n$ stations attached to a channel.
\end{theorem}

\begin{proof}
We investigate how  the adversary can organize an execution to maximize the delay of some packet.
The discrete time line underlying an execution is partitioned into double segments of eight consecutive rounds.
It is possible to activate three stations in a segment with one packet per activated station such that algorithm  \textsc{Quadruple-Round} spends seven rounds processing the double segment; see Anantharamu and Chlebus~\cite{AnantharamuC15} for similar arguments. 
This can be iterated with more triples of packets injected into the same stations in the double segment.
Finally, one more silent round verifies that the double segment is clear.
This creates the worst-case time overhead per the number of injected packets, and is the scheme repeatedly used in the argument on the level of executions.

A specific execution we present consists of two parts.
In the first part, the adversary works to make all the active stations store as many packets in total as possible.
This is accomplished by activating three stations in a segment, whenever there is an opportunity to do this, and such that the algorithm needs seven rounds to process such triples of packets.
We keep injecting into these stations, three packets per the same stations,  when the adversarial model allows to inject three packets, or a multiple of three, and when activating three new stations is not possible because fewer than three stations are passive. 
After the goal of the first part is accomplished, a station is activated with a new dedicated packet that we want to delay as much as possible.
This is accomplished by injecting at full power only into stations that were activated prior to the dedicated packet's station.

The first part begins by the size of the set of acting stations growing with the goal to make all the available $n$ stations active.
This is feasible because injecting a packet at every other round is sufficient to maintain the current size of the queue, in the sense that a new station gets  activated immediately after a station leaves the queue.
We define a \emph{phase} to be a contiguous segment of rounds during which  $n$ stations become active. 
In particular, if all $n$ stations are active simultaneously and each of them holds one packet, then a phase takes $\frac{8}{3}n$ rounds.
This is because there are $\frac{n}{3}$ segments, with three stations activated in a segment, each taking $8$ rounds to clear.
Once the $n$ stations are active, injecting with rate $\frac{3}{8}$ is sufficient to maintain this state in the ad-hoc channel, and with rate less that $\frac{3}{7}$ in the perpetual one.

Let $p$ denote the maximum number of packets accrued in the first part, disregarding burstiness, so the true maximum number is $m=p+\beta$.
The number of double segments in a phase is $\frac{n}{3}$, and this is the number of silent rounds resulting in clearing the double segments.
When the first part of an execution is in an equilibrium state, the following equality  is satisfied: 
\[
p = \rho \cdot \Bigl(\frac{1}{3}n + \frac{7}{3} p\Bigr)
\ .
\]
It  can be solved for $p$ to give $m$: $m= \frac{\rho}{3-7\rho}\cdot n+\beta$.
The delay of a packet injected now can be extended to the following value
\[
\Bigl(\frac{7}{3} m +\frac{1}{3} n \Bigr) 
\Bigl(1+\frac{7}{3}\rho + \Bigl(\frac{7}{3}\rho\Bigr)^2+\cdots \Bigr)
=
\frac{7m+n}{3-7\rho}
\ .
\] 
The claimed bound is obtained from this by algebra.
\end{proof}

\section{Simulations of Randomized Adversaries}

\label{sec:simulations-randomized-adversaries}

We present outcomes of experiments carried out as simulations of randomized adversaries.
The adversaries are specified by bucket processes, as presented in Section~\ref{sec:adversarial-packet-injection}.
A bucket process determines the number of packets generated in a step but does not determine the stations into which the packets get injected.
We specify a natural way to select stations for injecting packets that is consistent with the $1$-activating constraint.
Let us consider a $(\rho,\beta)$ randomized adversary injecting packets into the stations of a perpetual channel with $n$ stations.
In a round, first the number of packets~$j$ to inject is determined.
Next, the adversary selects one station uniformly at random from the set of passive stations as \emph{virtually active}, subject to the proviso that if all the stations are active then there will be no virtually active station.
The active nodes along with a virtually active (if any) node make the set of \emph{eligible nodes}  for this round.
For each of the new $j$ packets,  an eligible station is assigned to it independently and uniformly at random and the packet is injected into it.

A simulating execution proceeds by performing a broadcast algorithm and injecting packets subject to the type $(\rho,\beta)$ of a considered adversary; in our  experiments $\beta=10$.
The execution is partitioned into \emph{stages}, such that a stage concludes with  the average packet delay of a batch of some~$K$ packets; this parameter is set to $K=5000$ in our experiments.
A stage begins by designating $K$ consecutively generated  packets as \emph{marked} and concludes when all the marked packets have been heard on the channel.
As the execution proceeds, packets are generated continuously, but if they are not marked then their delay is not measured.
As a stage ends, the average packet delay is determined by the $K$ marked packets.

We say that an average packet latency \emph{stabilizes} when the relative average packet latencies in any two stages in four consecutive stages differ by less than~5\%.
If the average packet delay stabilizes then it is considered as the output of the experiment, otherwise the channel is considered unstable. 
This assessment is conservative with respect to stability, in  that some latency may be recorded even if the system is unstable in the stochastic sense. 
The latencies obtained in simulation vary significantly among the simulated algorithms, so we depict them on exponential vertical scales in the charts.


\begin{figure}[t]
\begin{center}

\includegraphics[width=350pt]{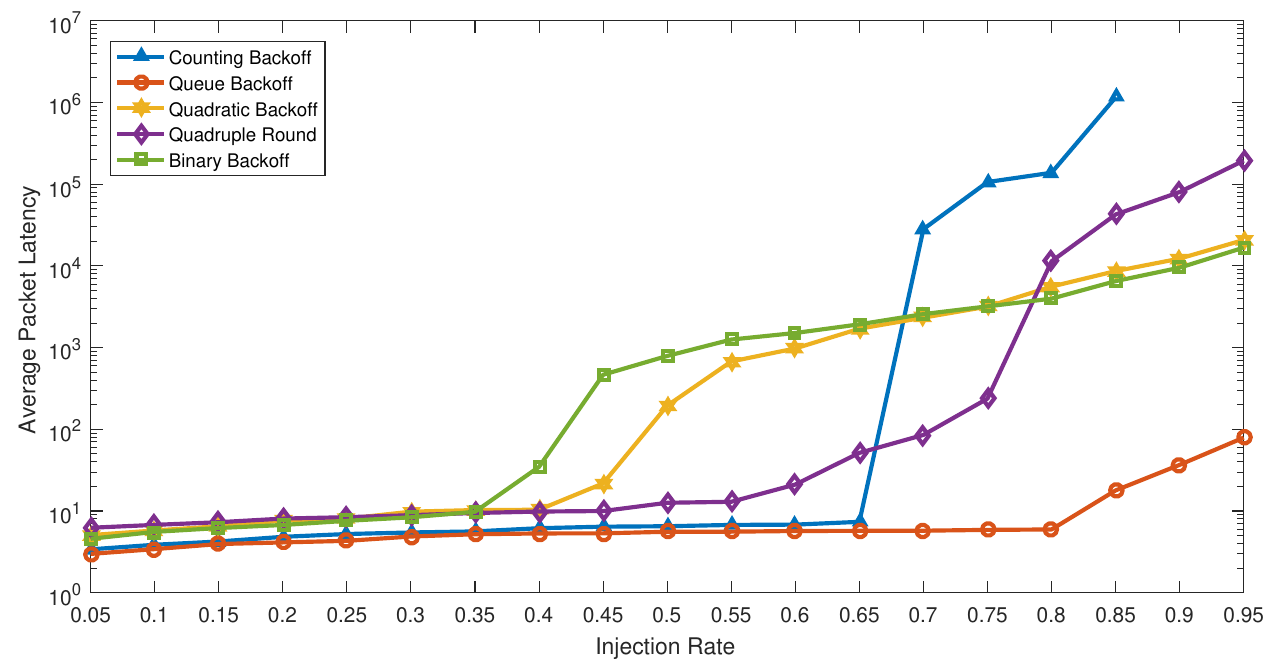}

\FFF

\parbox{\captionwidth}{
\caption{\label{fig:first}
A comparison of ad-hoc and backoff algorithms with respect to the average packet latency on $n=10$ stations for the full range of injection rates.}}
\end{center}
\end{figure}


\begin{figure}[t]
\begin{center}
\includegraphics[width=350pt]{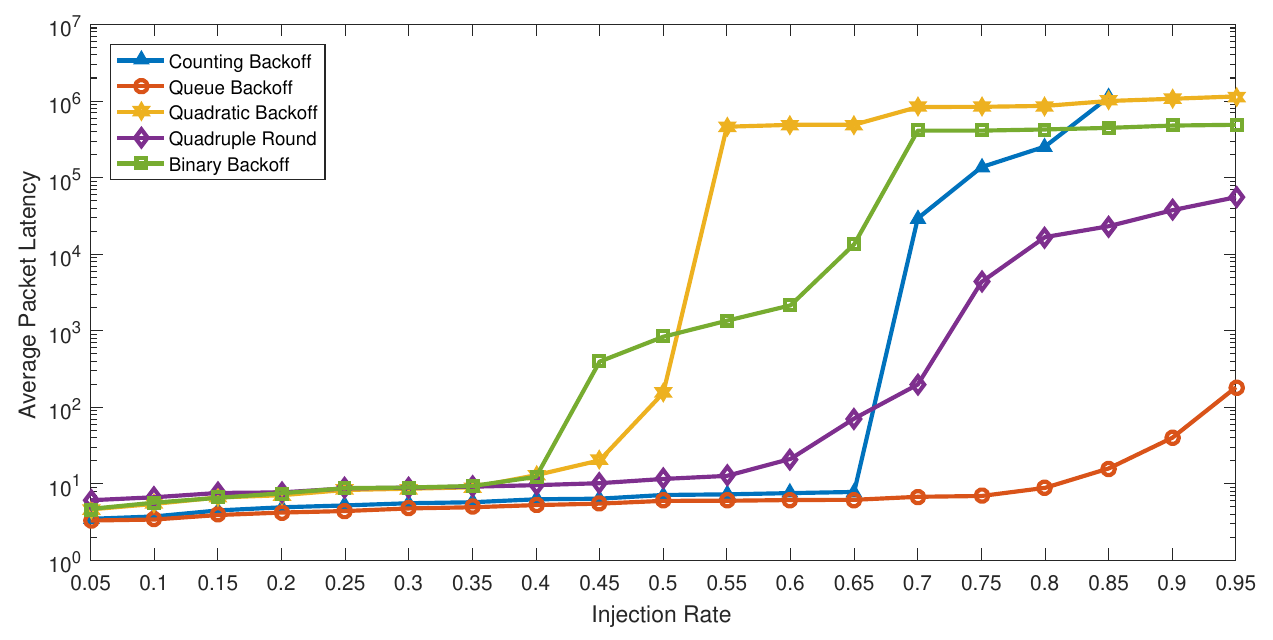}
\FFF

\parbox{\captionwidth}{
\caption{\label{fig:second}
A comparison of ad-hoc and backoff algorithms with respect to the average packet latency on $n=250$ stations for the full range of injection rates.}}
\end{center}
\end{figure}


\begin{figure}[t]
\begin{center}
\includegraphics[width=350pt]{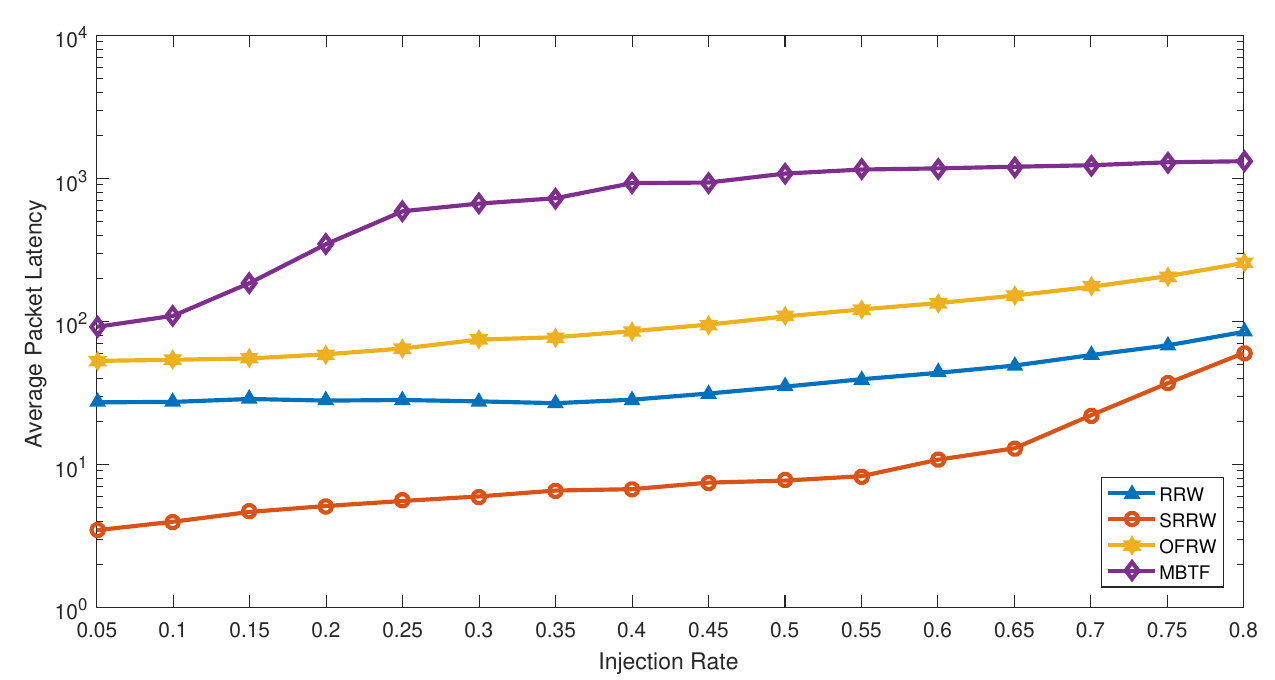}

\FFF
\parbox{\captionwidth}{
\caption{\label{fig:third}
A comparison of the average packet latency of token algorithms on $n=50$ stations for injection rates in the range $[0.05,0.80]$.}}
\end{center}/Users/bchlebus/My Drive/R/1-Current/1-Active/0-adv-broad-MAC/arXiv/2/Charts/RR250L-cropped.pdf
\end{figure}


\begin{figure}[t]
\begin{center}
\includegraphics[width=350pt]{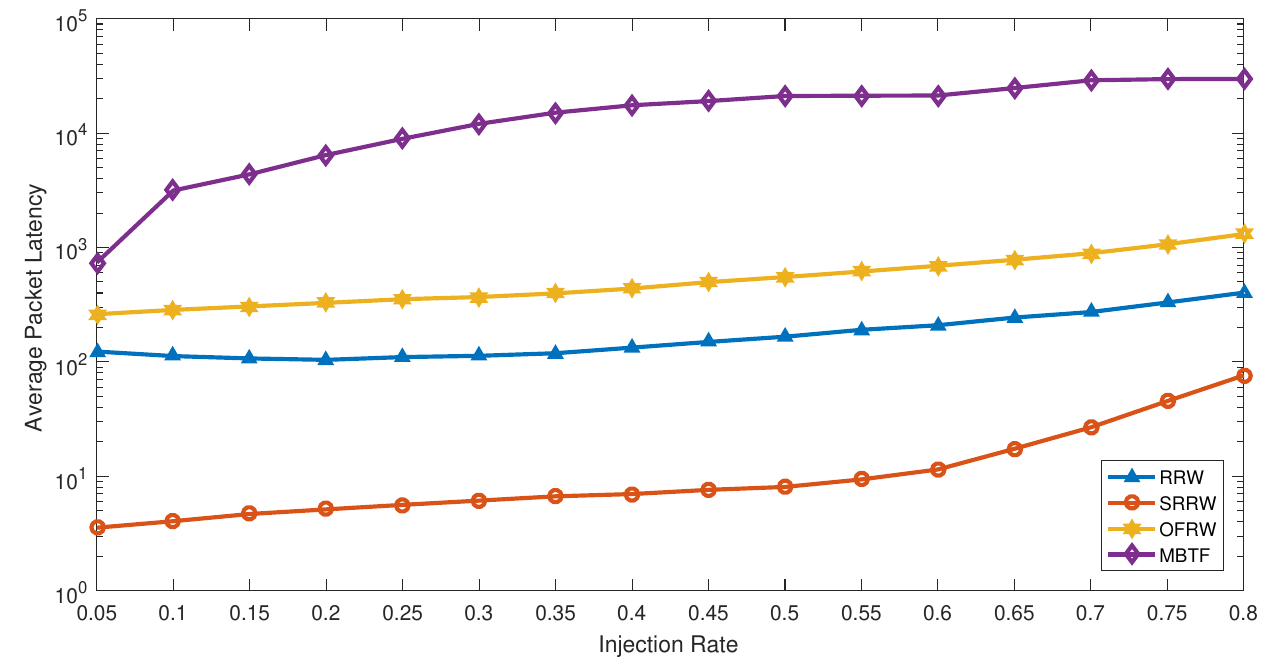}

\FFF
\parbox{\captionwidth}{
\caption{\label{fig:fourth}
A comparison of the  average packet latency of token algorithms on $n=250$ stations for injection rates in the range $[0.05,0.80]$.}}
\end{center}
\end{figure}


\begin{figure}[t]

\begin{center}
\includegraphics[width=350pt]{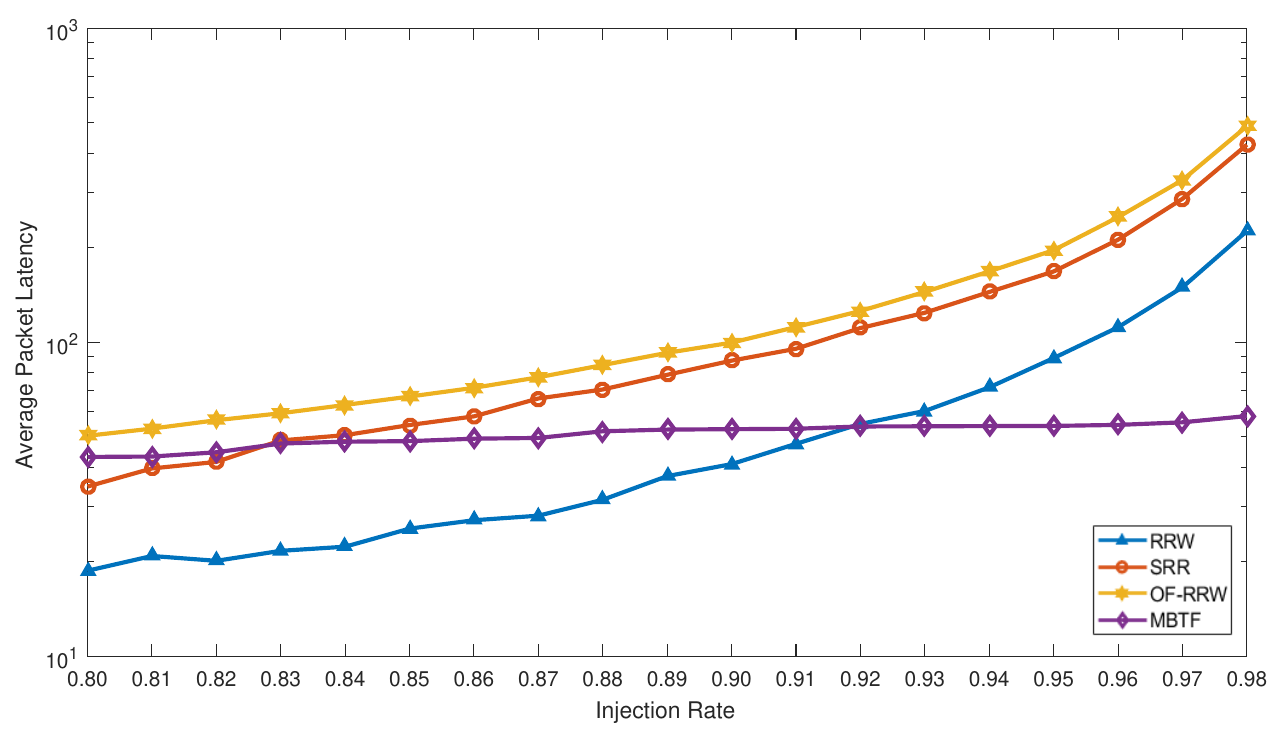}

\FFF

\parbox{\captionwidth}{
\caption{\label{fig:fifth}
A comparison of the  average packet latency of token algorithms on $n=10$ stations for injection rates in the range $[0.80 , 0.98]$.
}}
\end{center}
\end{figure}


\begin{figure}[t]
\begin{center}
\includegraphics[width=350pt]{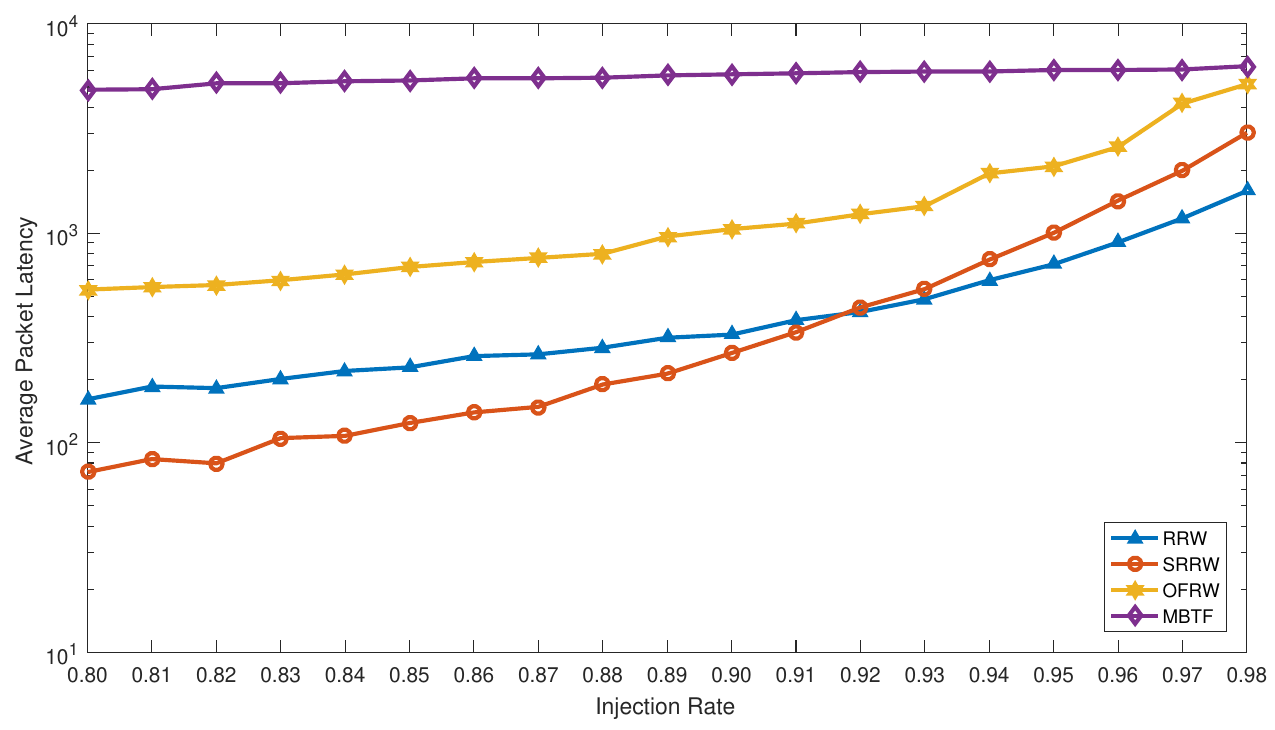}

\FFF
\parbox{\captionwidth}{
\caption{\label{fig:sixth}
A comparison of the  average packet latency of token algorithms on $n=100$ stations for  injection rates in the range $[0.80 , 0.98]$.}}
\end{center}
\end{figure}

Figures~\ref{fig:first} and~\ref{fig:second} present a comparison of ad-hoc deterministic algorithms and acknowledgement-based randomized algorithms in the whole spectrum of injection rates and for two sample sizes of the system: one relatively small and the other larger.
Latencies turn out to be comparable for small injection rates but vary significantly when injection rates approach~$1$.
In particular, the latencies for the algorithm \textsc{Counting-Backoff} do not stabilize in some range, which is consistent with this algorithm's instability in the worst-case sense for injection rates greater than~$\frac{1}{2}$.
The algorithm \textsc{Queue-Backoff} outperforms all the other algorithms. Individual comparisons depend on the number of stations; in particular, \textsc{Quadruple-Round} may outperform both backoff algorithms, or vice versa. 
The randomized algorithms \textsc{Binary-Exponential-Backoff} and \textsc{Quadratic-Backoff} have the largest latencies.
The mutual comparison of these two algorithms demonstrates the superiority of \textsc{Binary-Exponential-Backoff} as compared to  \textsc{Quadratic-Backoff}.
A theoretical comparison of stability of exponential and polynomial backoff algorithms was given by H\aa stad et al.~\cite{HastadLR-SICOMP96}.

Figures~\ref{fig:third} and~\ref{fig:fourth} show charts of experiments comparing token algorithms for injection rates up to~$0.8$, which is a range of moderate magnitude.
The relative performance is preserved for the two different sizes of the system.
The algorithm SRR is the most efficient and MBTF is the least efficient in both cases.

Figures~\ref{fig:fifth} and~\ref{fig:sixth} present experiments of the token algorithms for a range of large injection rates close to~$1$,  and for two sample sizes of the system: one relatively small and the other larger.
The best choice for large systems and low injection rates may be algorithm \textsc{SRR}.
Another observation is that \textsc{RRW} outperforms \textsc{OF-RRW} in this adversarial model, which is consistent with the analysis of the individual-injection rates adversaries in Section~\ref{sec:individual}, and is the opposite of the worst-case bounds given in Anantharamu et al.~\cite{AnantharamuCKR-JCSS19}.
The algorithm~\textsc{MBTF} may have the smallest latency for large injection rates, but when the number of stations is sufficiently small; this is consistent with the worst-case bound showed  in Anantharamu et al.~\cite{AnantharamuCKR-JCSS19}, which is a function of~$n^2$.

\section{Conclusion}

\label{sec:conclusion}

We presented a comprehensive view of broadcasting on adversarial multiple-access channels by deterministic distributed algorithms.
The given results demonstrate that performance bounds of broadcast algorithms depend on particular aspects of the broadcasting environment, which is defined by many components.
The featured examples include comparing perpetual channels versus ad-hoc ones, and comparing the models of an unconstrained leaky-bucket injection rate versus an individually-constrained injection rate. 
 
The leaky-bucket adversaries allow to investigate stability and latency of broadcasting on multiple-access channel without stochastic assumptions on packet injection.
The performance bounds of deterministic algorithms for such adversarial traffic are in the worst-case sense because the adversarial packet injection is constrained by specific upper bounds on how many packets can be injected in a bounded time interval as a function of the length of the interval.
The  prospect of  simulating leaky-bucket adversaries, for example, in order to verify some performance bounds, faces immediate challenges. 
One is that the constrains on the adversaries are in the form of an upper bound of how many packets the adversary may inject exercising the ``full power of packet generation,'' while this behavior does not necessarily result in the worst-case performance.
Another is that there are exponentially many strategies of the adversary to inject packets in a time interval, as a function of the length of the interval.

We propose a randomized model of adversarial packet injection, which allows for arbitrarily bursty traffic to occur with positive probability.
This model is amenable to simulations and the average performance of broadcasting algorithms can be discovered via simulated experiments. 
The examples of such simulations that we give indicate a complex landscape of behaviors of broadcast algorithms, with performance bounds depending on the number of stations in a perpetual channel as well as on the injection rates.
These experiments involve two popular randomized backoff protocols, binary-exponential and quadratic ones.
The most interesting injection rates are those close to the ultimate injection rate~$1$, as they differentiate the algorithms most.
Such differentiation is not necessarily simple or conclusive, as the relative packet latency of two broadcast algorithms may depend on the number of stations attached to the channel. 
Deriving quantitative performance bounds on the expected packet latency and queue sizes of the studied broadcast algorithms for the model of randomized adversaries is an interesting  direction of future research.

\bibliographystyle{plain}

\bibliography{comparative-mac}

\begin{thebibliography}{10}

\bibitem{AgrawalBFGY20}
Kunal Agrawal, Michael~A. Bender, Jeremy~T. Fineman, Seth Gilbert, and Maxwell
  Young.
\newblock Contention resolution with message deadlines.
\newblock In {\em Proceedings of the $32$nd {ACM} Symposium on Parallelism in
  Algorithms and Architectures (SPAA)}, pages 23--35. {ACM}, 2020.

\bibitem{AielloKOR-JCSS00}
William Aiello, Eyal Kushilevitz, Rafail Ostrovsky, and Adi Ros{\'e}n.
\newblock Adaptive packet routing for bursty adversarial traffic.
\newblock {\em Journal of Computer and System Sciences}, 60(3):482--509, 2000.

\bibitem{Al-AmmalGM-TCSy01}
Hesham Al-Ammal, Leslie~Ann Goldberg, and Philip~D. MacKenzie.
\newblock An improved stability bound for binary exponential backoff.
\newblock {\em Theory of Computing Systems}, 34(3):229--244, 2001.

\bibitem{AldawsariCK19}
Bader~A. Aldawsari, Bogdan~S. Chlebus, and Dariusz~R. Kowalski.
\newblock Broadcasting on adversarial multiple access channels.
\newblock In {\em Proceedings of the $18$th {IEEE} International Symposium on
  Network Computing and Applications (NCA)}, pages 1--4. {IEEE}, 2019.

\bibitem{AnantharamuC15}
Lakshmi Anantharamu and Bogdan~S. Chlebus.
\newblock Broadcasting in ad hoc multiple access channels.
\newblock {\em Theoretical Computer Science}, 584:155--176, 2015.

\bibitem{AnantharamuCKR-JCSS19}
Lakshmi Anantharamu, Bogdan~S. Chlebus, Dariusz~R. Kowalski, and Mariusz~A.
  Rokicki.
\newblock Packet latency of deterministic broadcasting in adversarial multiple
  access channels.
\newblock {\em Journal of Computer and System Sciences}, 99:27--52, 2019.

\bibitem{AnantharamuCR-TCS17}
Lakshmi Anantharamu, Bogdan~S. Chlebus, and Mariusz~A. Rokicki.
\newblock Adversarial multiple access channels with individual injection rates.
\newblock {\em Theory of Computing Systems}, 61(3):820--850, 2017.

\bibitem{AndertonCY21}
William~C. Anderton, Trisha Chakraborty, and Maxwell Young.
\newblock Windowed backoff algorithms for wifi: theory and performance under
  batched arrivals.
\newblock {\em Distributed Computing}, 34(5):367--393, 2021.

\bibitem{AndrewsAFLLK-JACM01}
Matthew Andrews, Baruch Awerbuch, Antonio Fern{\'a}ndez, Frank~Thomson
  Leighton, Zhiyong Liu, and Jon~M. Kleinberg.
\newblock Universal-stability results and performance bounds for greedy
  contention-resolution protocols.
\newblock {\em Journal of the ACM}, 48(1):39--69, 2001.

\bibitem{AndrewsFGZ-JACM05}
Matthew Andrews, Antonio Fern{\'a}ndez, Ashish Goel, and Lisa Zhang.
\newblock Source routing and scheduling in packet networks.
\newblock {\em Journal of the ACM}, 52(4):582--601, 2005.

\bibitem{BanicescuCGY2024}
Ioana Banicescu, Trisha Chakraborty, Seth Gilbert, and Maxwell Young.
\newblock A survey on adversarial contention resolution.
\newblock {\em CoRR}, arXiv:2403.03876, 2024.

\bibitem{BenderFHKL-SPAA05}
Michael~A. Bender, Martin Farach-Colton, Simai He, Bradley~C. Kuszmaul, and
  Charles~E. Leiserson.
\newblock Adversarial contention resolution for simple channels.
\newblock In {\em Proceedings of the $17$th ACM Symposium on Parallel
  Algorithms and Architectures (SPAA )}, pages 325--332, 2005.

\bibitem{BenderFGY19}
Michael~A. Bender, Jeremy~T. Fineman, Seth Gilbert, and Maxwell Young.
\newblock Scaling exponential backoff: Constant throughput, polylogarithmic
  channel-access attempts, and robustness.
\newblock {\em Journal of the {ACM}}, 66(1):6:1--6:33, 2019.

\bibitem{BenderKKP20}
Michael~A. Bender, Tsvi Kopelowitz, William Kuszmaul, and Seth Pettie.
\newblock Contention resolution without collision detection.
\newblock In {\em Proccedings of the $52$nd Annual {ACM} {SIGACT} Symposium on
  Theory of Computing (STOC)}, pages 105--118. {ACM}, 2020.

\bibitem{BenderKPY18}
Michael~A. Bender, Tsvi Kopelowitz, Seth Pettie, and Maxwell Young.
\newblock Contention resolution with constant throughput and log-logstar
  channel accesses.
\newblock {\em {SIAM} Journal on Computing}, 47(5):1735--1754, 2018.

\bibitem{BergerKS14}
Daniel~S. Berger, Martin Karsten, and Jens~B. Schmitt.
\newblock On the relevance of adversarial queueing theory in practice.
\newblock In {\em Proceedings of the International Conference on Measurement
  and Modeling of Computer Systems (SIGMETRICS)}, pages 343--354. {ACM}, 2014.

\bibitem{Bianchi00}
Giuseppe Bianchi.
\newblock Performance analysis of the {IEEE} 802.11 distributed coordination
  function.
\newblock {\em IEEE Journal on Selected Areas in Communications}, 18:535 --
  547, 2000.

\bibitem{BorodinKRSW-JACM01}
Allan Borodin, Jon~M. Kleinberg, Prabhakar Raghavan, Madhu Sudan, and David~P.
  Williamson.
\newblock Adversarial queuing theory.
\newblock {\em Journal of the ACM}, 48(1):13--38, 2001.

\bibitem{ChangJP19}
Yi{-}Jun Chang, Wenyu Jin, and Seth Pettie.
\newblock Simple contention resolution via multiplicative weight updates.
\newblock In {\em Proceedings of the $2$nd Symposium on Simplicity in
  Algorithms (SOSA)}, volume~69 of {\em {OASICS}}, pages 16:1--16:16. Schloss
  Dagstuhl - Leibniz-Zentrum f{\"{u}}r Informatik, 2019.

\bibitem{ChenJZ21}
Haimin Chen, Yonggang Jiang, and Chaodong Zheng.
\newblock Tight trade-off in contention resolution without collision detection.
\newblock In {\em Proceedings of the $41$st {ACM} Symposium on Principles of
  Distributed Computing (PODC)}, pages 139--149. {ACM}, 2021.

\bibitem{Chlebus-chapter-2001}
Bogdan~S. Chlebus.
\newblock Randomized communication in radio networks.
\newblock In Panos~M. Pardalos, Sanguthevar Rajasekaran, John~H. Reif, and Jose
  D.~P. Rolim, editors, {\em Handbook of Randomized Computing}, volume~I, pages
  401--456. Kluwer Academic Publishers, 2001.

\bibitem{ChlebusHJKK-SPAA19}
Bogdan~S. Chlebus, Elijah Hradovich, Tomasz Jurdzi\'nski, Marek Klonowski, and
  Dariusz~R. Kowalski.
\newblock Energy efficient adversarial routing in shared channels.
\newblock In {\em Proceedings of the $31$st ACM Symposium on Parallelism in
  Algorithms and Architectures (SPAA)}, pages 191--200. ACM, 2019.

\bibitem{ChlebusKR-DC09}
Bogdan~S. Chlebus, Dariusz~R. Kowalski, and Mariusz~A. Rokicki.
\newblock Maximum throughput of multiple access channels in adversarial
  environments.
\newblock {\em Distributed Computing}, 22(2):93--116, 2009.

\bibitem{ChlebusKR-TALG12}
Bogdan~S. Chlebus, Dariusz~R. Kowalski, and Mariusz~A. Rokicki.
\newblock Adversarial queuing on the multiple access channel.
\newblock {\em ACM Transactions on Algorithms}, 8(1):5:1--5:31, 2012.

\bibitem{GoldbergJKP04}
Leslie~Ann Goldberg, Mark Jerrum, Sampath Kannan, and Mike Paterson.
\newblock A bound on the capacity of backoff and acknowledgment-based
  protocols.
\newblock {\em SIAM Journal on Computing}, 33(2):313--331, 2004.

\bibitem{HastadLR-SICOMP96}
Johan H{\aa }stad, Frank~Thompson Leighton, and Brian Rogoff.
\newblock Analysis of backoff protocols for multiple access channels.
\newblock {\em SIAM Journal on Computing}, 25(4):740--774, 1996.

\bibitem{HradovichKK21}
Elijah Hradovich, Marek Klonowski, and Dariusz~R. Kowalski.
\newblock New view on adversarial queueing on {MAC}.
\newblock {\em {IEEE} Communication Letters}, 25(4):1144--1148, 2021.

\bibitem{JiangZ22}
Yonggang Jiang and Chaodong Zheng.
\newblock Robust and optimal contention resolution without collision detection.
\newblock In {\em Proceedings of the $34$th {ACM} Symposium on Parallelism in
  Algorithms and Architectures (SPAA)}, pages 107--118. {ACM}, 2022.

\bibitem{KwakSM05}
Byung-Jae Kwak, Nah-Oak Song, and Leonard~E. Miller.
\newblock Performance analysis of exponential backoff.
\newblock {\em IEEE/ACM Transactions on Networking}, 13(2):343--355, 2005.

\bibitem{MetcalfeB76}
Robert~M. Metcalfe and David~R. Boggs.
\newblock Ethernet: {D}istributed packet switching for local computer networks.
\newblock {\em Communications of the ACM}, 19(7):395--404, 1976.

\bibitem{MitzenmacherUpfal-book17}
Michael Mitzenmacher and Eli Upfal.
\newblock {\em Probability and Computing}.
\newblock Cambridge University Press, {S}econd edition, 2017.

\bibitem{RosenT-TCSy06}
Adi Ros{\'e}n and Michael~S. Tsirkin.
\newblock On delivery times in packet networks under adversarial traffic.
\newblock {\em Theory of Computing Systems}, 39(6):805--827, 2006.

\bibitem{ZhouCY22}
Qian~M. Zhou, Alice Calvert, and Maxwell Young.
\newblock Singletons for simpletons revisiting windowed backoff with {C}hernoff
  bounds.
\newblock {\em Theoretical Computer Science}, 909:39--53, 2022.

\end{thebibliography}

\end{document}